\documentclass[twocolumn]{autart}    

\usepackage{graphicx}          
\usepackage[svgnames,dvipsnames]{xcolor}

\usepackage{amssymb,amsmath}

\usepackage{url}

\usepackage[figuresright]{rotating}
\newtheorem{theorem}{Theorem}[section]

\newtheorem{proposition}{Proposition}[theorem]
\newtheorem{lemma}[theorem]{Lemma}
\newtheorem{example}[theorem]{Example}

\usepackage[figuresright]{rotating}

\usepackage[sort]{natbib}

\begin{document}

\begin{frontmatter}

\title{A modular approach to achieve multistationarity using AND-gates} 

\thanks[footnoteinfo]{A.VC. received the following funding. Grant for PUI institutions, American Mathematical Society and the Simons Foundation. NSF Emergent Math in Biology grant (number 2424634)}

\author[Dayton]{Alan Veliz-Cuba}\ead{avelizcuba1@udayton.edu},    
\author[Dayton]{Zeyu Wang}

\address[Dayton]{University of Dayton, Dayton, Ohio, USA}  


\begin{abstract}                          
Systems of differential equations have been used to model biological systems such as gene and neural networks. A problem of particular interest is to understand the number of stable steady states. 
Here we propose conjunctive networks (systems of differential equations equations created using AND gates) to achieve any desired number of stable steady states. Our approach uses combinatorial tools to  predict the number of stable steady states from the structure of the wiring diagram. 
Furthermore, AND gates have been successfully engineered by experimentalists for gene networks, so our results provide a modular approach to design gene networks that achieve arbitrary number of phenotypes. 

\end{abstract}

\end{frontmatter}

\section{Introduction}\label{sec:Intro}


Understanding the stable steady states of systems of ordinary differential equations (ODEs) is a problem that arises when studying different biological problems such as bistability, cell types, multistationarity, memory formation. A particular problem of interest is to be able to predict the number of stable steady states efficiently and to be able to create an ODE given a certain number of desired stable steady states.

A typical system of differential equations used in modeling has the form

\[\frac{dx_i}{dt}=F_i(\mathbf{x}) - \beta x_i,\]

where $i=1,\ldots,N,$ $\mathbf{x}=(x_1,\ldots,x_n)$,  and $x_i$ models the concentration or activity level of some chemical/gene/neuron. The parameter $\beta$ corresponds to natural decay or dilution, and the function $F_i$ encodes the way in which $x_i$ depends on all other variables. The function $F_i$ is typically constructed using nonlinear functions such as the Hill function, 

\[H(z)=\frac{z^n}{\theta^n+z^n},\]

 or an algebraic combination of Hill functions. Here $\theta$ is the activation threshold and $n$ is the Hill coefficient that determines how steep the activation occurs around the threshold (Fig.~\ref{fig:Hillfunction}). 

\begin{figure}[h]
\centerline{ \hbox{ \includegraphics[width=7.5cm]{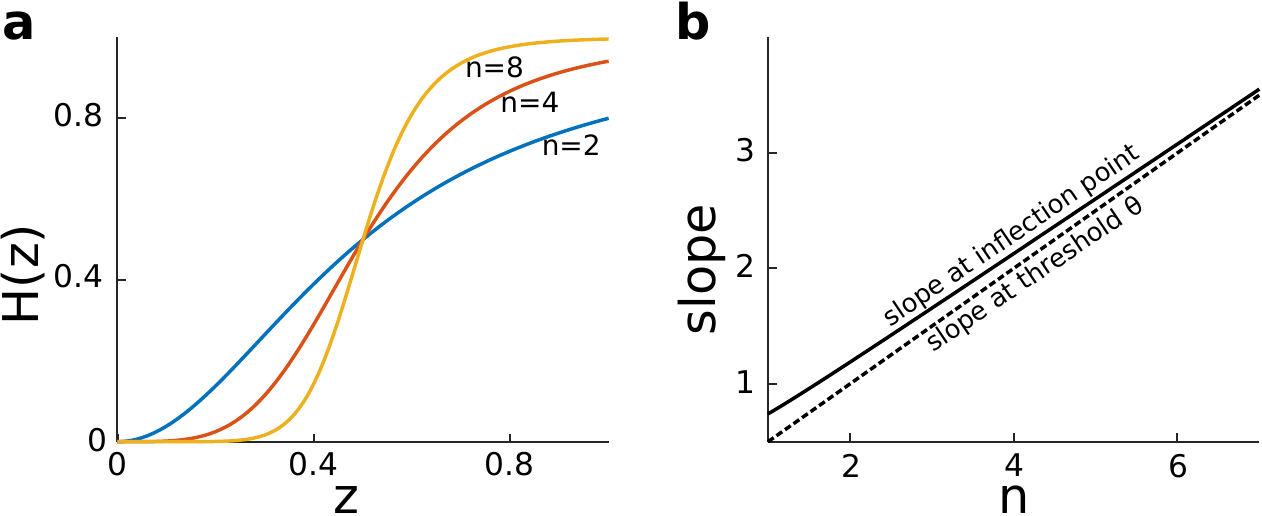}
}}\vspace{-7pt}
 \caption{
Hill function. 
(\textbf{a}) Plots of Hill functions for different values of the Hill coefficient ($\theta=1/2$). 
(\textbf{b}) Value of the slope as a function of the Hill coefficient ($\theta=1/2$). Gray: value of the slope at the threshold $\theta$; black: value of the slope at the inflection point, $\theta \sqrt[n]{\frac{n-1}{n+1}}$. Note that as $n$ increases, the inflection point approaches $\theta$ and the corresponding slopes become indistinguishable.
}
 \label{fig:Hillfunction}
\end{figure}

In this manuscript we focus on the case where the functions $F_i$ are constructed using AND gates. Intuitively, an AND gate consists of a function where all the inputs have to be large for the output to be high (Fig.~\ref{fig:ANDgate}). There are several cases where biological regulation behaves like or can be approximated  by an AND gate \cite{VelizStigler:lacop,Veliz:AND_NOT_networks,doi:10.1126/sciadv.adj0822} and furthermore, AND gates have been engineered in laboratories \cite{Shis26032013,C4CC10047F}. Also, Boolean networks constructed with AND gates have been studied in the past \cite{CBN,MaxNet,Gao7963728,GAO20188,WEISS201856}, and the technique of using the network structure to predict dynamical features may potentially be applicable to their continuous counterpart.
To model AND gates we use products of Hill functions, 

\[
\displaystyle F_i(\mathbf{x})=\alpha \prod_{k\in I_i} \frac{x_k^n}{\theta^n+x_k^n},
\]
where $\alpha$ is the maximum activation and $I_i$ denotes the set of variables that affect variable $x_i$. 


\begin{figure}[h]
\centerline{ \hbox{ 
 \includegraphics[width=8cm]{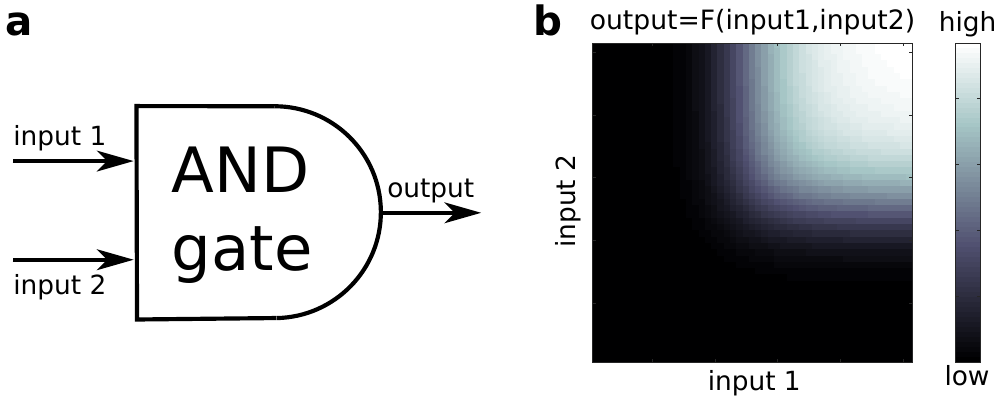}
}}\vspace{-7pt}
 \caption{
Schematic of an AND gate. 
(\textbf{a}) Circuit representation of an AND gate in two variables. The output will be high when both inputs are high. 
(\textbf{b}) Qualitative behavior of an AND gate. The heatmap shows that if either of the inputs is below a threshold, then the output takes a low value. If both inputs are above a threshold, then the output takes a high value. Note the nonlinearity of function $F$.
}
 \label{fig:ANDgate}
\end{figure}

To make the presentation more transparent, we will assume that all variables have the same activation threshold $\theta=1/2$, as well as the same Hill coefficient, $n$. Furthermore, rescaling time and $x_i$, we can assume without loss of generality that $\alpha=\beta=1$. Note that with these assumptions, the set $[0,1]^N$ is invariant. We call these systems of ordinary differential equations \emph{conjunctive networks}. 

{
We now summarize our main results. In Theorem \ref{thm:sc} we show that if the wiring diagram is strongly connected then the network has exactly two steady states. In Theorem \ref{thm:antichain} we generalize the previous theorem and give the exact number of stable steady states of conjunctive networks with arbitrary wiring diagram. In Theorem \ref{thm:exist_conj_net} we show a constructive way to design networks with arbitrary multistability and in Theorem \ref{thm:factors} we improve the construction by showing how to reduce the number of variables needed. Proposition \ref{prop:basin} and Theorem \ref{thm:global} in essence state that the basin of attractions for the steady states are not small and that all trajectories converge to one of the steady states predicted by our approach.
}

\section{Conjunctive ODEs}\label{sec:conjODEs}
A conjunctive network is a system of differential equations of the form 
\begin{equation}\label{eq:conj}
\frac{dx_i}{dt}=\prod_{k\in I_i} \frac{x_k^n}{\theta^n+x_k^n} - x_i,
\end{equation}
where $I_i\subseteq \{1,\ldots,N\}$ for each $i=1,\ldots,N$.

Given a conjunctive network, we construct a directed graph with vertices $i\in\{1,\ldots,N\}$ (or $x_i$) and edges $j\rightarrow i$ if and only if $j\in I_i$. We call this directed graph, \emph{wiring diagram} or \emph{dependency graph}. The wiring diagram describes the dependency between variables and its structure will be key in determining the number of stable steady states. In the context of modularity \cite{wheeler2024modular,kadelka2023modularity}, a conjunctive network can be seen as a network made up of simple modules (AND gates), and the way we connect these modules can potentially permit the network achieve complex dynamics. 
The sigmoidality of the regulation functions is governed by the Hill coefficient $n$.

\textbf{Remark.}
Unlike studies that rely on the idealized piecewise-linear limit ($n\rightarrow \infty$), we analyze the system for finite $n$. Specifically, we establish that our results hold for any $n\geq n_0$, where the existence of this bound is guaranteed by \cite{Veliz:BNODE2012,veliz2014piecewise}. 
To maintain conciseness, we adopt the phrase ``for large but finite $n$'' to refer specifically to this regime $n\geq n_0$. In particular, results about conjunctive networks for discrete systems are not guaranteed a priori.

{
The stationary points of interest in this manuscript are \emph{asymptotically stable steady states}, that is, steady states $\mathbf{x}$ such that (1) for every $\epsilon>0$ there is $\delta>0$ such that for every solution $\mathbf{z}(t)$, $|\mathbf{z}(0)-\mathbf{x}|<\delta$ guarantees $|\mathbf{z}(t)-\mathbf{x}|<\epsilon$ for all $t$, and (2)  every trajectory that starts sufficiently close to $\mathbf{x}$ will converge to $\mathbf{x}$. For brevity, in the rest of the manuscript we refer to asymptotically stable steady states simply as  \emph{stable steady states} unless indicated otherwise. 
}

Also, if all entries of a steady state are positive, we say that the steady state is positive. 
When modeling biological networks, it is common to assume that the interactions are sigmoidal enough so that they behave qualitatively like Boolean gates and the behavior of the network depends on the qualitative features of the nonlinearities instead of the actual parameters \cite{MR2069236,Albert2010,Davidich:2008uq,Veliz:BNODE2012,VelizStigler:lacop,DavBor,Kauff2,veliz2014piecewise}.

\begin{example}\label{eg:1d}
Consider the 1-dimensional conjunctive network given by 
\[\frac{dx_1}{dt}=\frac{x_1^n}{\theta^n+x_1^n}-x_1.\] Its wiring diagram is shown in Fig.~\ref{fig:eg_1d}. For large but finite $n$, there are two stable steady states: $x_1=0$ and a positive stable steady state, Fig.~\ref{fig:eg_1d}b. Furthermore, as $n$ increases, the positive stable steady state approaches 1, Fig.~\ref{fig:eg_1d}c.
\end{example}

\begin{figure}[h]
\centerline{ \hbox{ 
 \includegraphics[width=8.3cm]{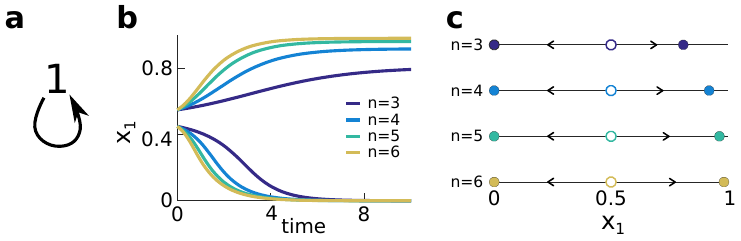}
}}\vspace{-7pt}
 \caption{
Dynamics of conjunctive  network in Example~\ref{eg:1d}.
(\textbf{a}) 
Wiring diagram.
(\textbf{b}) 
Behavior of solutions for different values of $n$. For initial conditions starting below 0.5, solutions converge to 0 and for initial conditions above 0.5, solutions converge to a positive stable steady state.
(\textbf{c})
Phase portrait of solutions for different values of $n$. The filled circles mark the stable steady states from panel (\textbf{b}) and the open circle denotes the unstable steady state. The positive stable steady state approaches 1 as $n$ increases. 
}
 \label{fig:eg_1d}
\end{figure}

\begin{example}\label{eg:2d}
Consider the 2-dimensional conjunctive network given by 
\begin{align*}
\frac{dx_1}{dt} & =\frac{x_1^n}{\theta^n+x_1^n}\frac{x_2^n}{\theta^n+x_2^n}-x_1,\\
\frac{dx_2}{dt} & =\frac{x_1^n}{\theta^n+x_1^n}-x_2.
\end{align*}
Its wiring diagram is shown in Fig.~\ref{fig:eg_2d}a. For large but finite $n$, there are two stable steady states: (0,0) and a positive stable steady state, Fig.~\ref{fig:eg_2d}b. Furthermore, as $n$ increases, the positive stable steady state approaches (1,1), Fig.~\ref{fig:eg_2d}c.
\end{example}

\begin{figure}[h]
\centerline{ \hbox{ 
 \includegraphics[width=8.3cm]{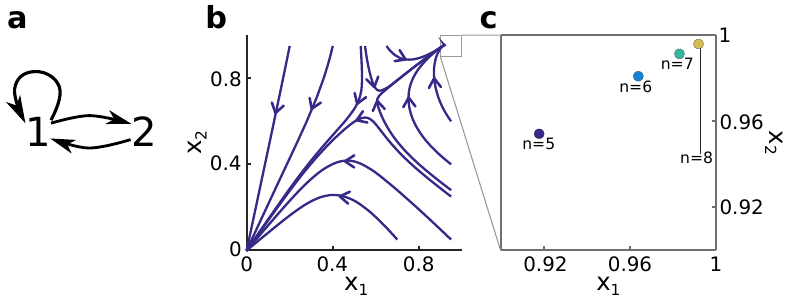}
}}\vspace{-7pt}
 \caption{
Dynamics of conjunctive  network in Example~\ref{eg:2d}.
(\textbf{a}) 
Wiring diagram.
(\textbf{b}) 
Phase portrait of conjunctive network for $n=5$. We can see that (up to a set of measure zero) all solutions converge to (0,0) or to the positive stable steady state. 
(\textbf{c})
Positive stable steady state for different values of $n$. Note that the positive stable steady state approaches (1,1) as $n$ increases. 
}
 \label{fig:eg_2d}
\end{figure}

\begin{example}\label{eg:3d}
Consider the 3-dimensional conjunctive network given by 
\begin{align*}
\frac{dx_1}{dt} & =\frac{x_2^n}{\theta^n+x_2^n}-x_1,
\\
\frac{dx_2}{dt} & =\frac{x_1^n}{\theta^n+x_1^n}-x_2,
\\
\frac{dx_3}{dt} & =\frac{x_1^n}{\theta^n+x_1^n}\frac{x_2^n}{\theta^n+x_2^n}\frac{x_3^n}{\theta^n+x_3^n}-x_3.
\end{align*}
Its wiring diagram is shown in Fig.~\ref{fig:eg_3d}a. For large but finite $n$, there are 3 stable steady states: (0,0,0), a stable steady state of the form $(x_1^*,x_2^*,0)$ where $x_1^*,x_2^*>0$, and a positive stable steady state, $(x_1^*,x_2^*,x_3^*)$ (Fig.~\ref{fig:eg_3d}b). Furthermore, as $n$ increases, the positive entries of these stable steady states approach 1 (Fig.~\ref{fig:eg_3d}c).
\end{example}

\begin{figure}[h]
\centerline{ \hbox{ 
 \includegraphics[width=8.3cm]{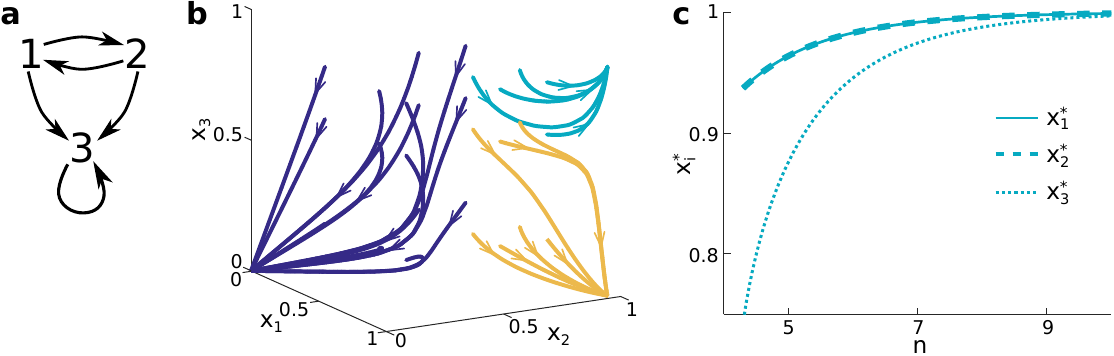}
}}\vspace{-7pt}
 \caption{
Dynamics of conjunctive  network in Example~\ref{eg:3d}.
(\textbf{a}) 
Wiring diagram.
(\textbf{b}) 
Phase portrait of conjunctive network for $n=5$. We can see that (up to a set of measure zero) all solutions converge to (0,0,0), to a positive stable steady state, or a stable steady state of the form $(x_1^*,x_2^*,0)$. 
(\textbf{c})
Values of the entries of the positive stable steady state as a function of $n$. The positive stable steady state converges to (1,1,1) as $n$ increases. The first two entries of the stable steady state of the form $(x_1^*,x_2^*,0)$ are identical to that of the positive stable steady state, so this stable steady state converges to $(1,1,0)$.
}
 \label{fig:eg_3d}
\end{figure}

We will use this family of differential equations to achieve arbitrary number of stable steady states. More precisely, if we want a network with exactly $m$ stable steady states, then we need to use the correct number of variables and the correct connectivity so that this conjunctive ODE has exactly $m$ stable steady states. In order to achieve this, we first need to be able to predict the number of stable steady states of a conjunctive network.

\section{Stable steady states of Conjunctive Networks}\label{sec:steady_states}

\subsection{Counting the number of stable steady states}

As observed in the previous section, the nonzero entries of stable steady states converge to 1 as $n$ increases. In the limit $n\rightarrow \infty$, the interaction between variables become Boolean functions, and then there is a one-to-one correspondence between the stable steady states of differential equations and stable steady states of the corresponding Boolean networks (which are binary strings). Such results have been proven in \cite{Veliz:BNODE2012,veliz2014piecewise} not only in the limit of $n\rightarrow \infty$, but also for finite $n$. For our purposes, the precise statement is the following.

\begin{theorem}\label{thm:previous}
Consider a conjunctive network in $N$ variables with Hill coefficient $n$ (Eq.~\ref{eq:conj}). For large but finite $n$, each stable steady state $\mathbf{x}_n$ is in $[0,1]^N$ and satisfies $\lim_{n\rightarrow \infty}\mathbf{x}_n\in\{0,1\}^N$. Also, there is at most one stable steady state in each of the $2^N$ regions obtained by cutting $[0,1]^N$ with the $N$ hyperplanes $x_i=\theta$.
\end{theorem}

For Example~\ref{eg:1d}, the theorem states that there is at most one stable steady state in $[0,\theta)$ and in $(\theta,1]$. We can see in Fig.\ref{fig:eg_1d}c that there is exactly one stable steady state in each interval. For Example~\ref{eg:2d}, there is at most one stable steady state in each of the four regions 
$[0,\theta)^2$, 
$[0,\theta)\!\times\! (\theta,1]$, 
$(\theta,1]\!\times\! [0,\theta)$, 
$(\theta,1]^2$. Indeed, in Fig.~\ref{fig:eg_2d}bc we can see that the region $[0,\theta)^2$ has the unique stable steady state (0,0), and $(\theta,1]^2$ has a unique stable steady state that converges to (1,1) as $n$ increases. The regions $[0,\theta)\!\times\! (\theta,1]$ and
$(\theta,1]\!\times\! [0,\theta)$ do not have any stable steady state. For Example~\ref{eg:3d}, there are 8 regions and from Fig.~\ref{fig:eg_3d} we see that only three of them have stable steady states: $[0,\theta)^3$ (note the blue trajectories), $(\theta,1]^3$  (note the cyan trajectories), and $(\theta,1]^2\times [0,\theta)$  (note the yellow trajectories).

We will now relate the structure of the wiring diagram with the number of stable steady states. To do that we need the following terminology. 

The wiring diagram of a conjunctive network is called \emph{strongly connected} if for any pair of vertices $i,j$, there is a directed path from $i$ to $j$. For example, the wiring diagrams in Fig.~\ref{eg:1d}a and Fig.~\ref{eg:2d}a are strongly connected. On the other hand, the wiring diagram in Fig.~\ref{eg:3d}a is not strongly connected, since there is no directed path from 3 to 1 (or from 3 to 2).

Using Theorem~\ref{thm:previous}, we prove the following.

\begin{theorem}\label{thm:sc}
Suppose the wiring diagram of  a conjunctive network is strongly connected. For large but finite $n$, there are exactly two stable steady states: the zero stable steady state and a positive stable steady state that approaches $(1,\ldots,1)$ as $n$ increases.
\end{theorem}
\textbf{Proof.}
Consider the function $F=(F_1,\ldots,F_N)$ given by $F_i(\mathbf{x})=\prod_{k\in I_i} \frac{x_k^n}{\theta^n+x_k^n}$ and note that Eq.~\ref{eq:conj} is simply $\mathbf{x}'=F(\mathbf{x})-\mathbf{x}$. For $\mathbf{x}$ to be a stable steady state we need to show that $F(\mathbf{x})=\mathbf{x}$ and that all eigenvalues of the Jacobian matrix of $F(\mathbf{x})-\mathbf{x}$ have negative real part.

Since $\frac{d}{dz} \frac{z^n}{\theta^n+z^n}=\frac{n\theta^n z^{n-1}}{(\theta^n+z^n)^2}$, it follows that the Jacobian matrix of $F(\mathbf{x})$, $J(\mathbf{x})$, is given by
\[ 
J_{ij}(\mathbf{x}) = 
\begin{cases}
0 , \text{ if $j \notin I_i$,}\\
%
\displaystyle\frac{n\theta^n x_j^{n-1}}{(\theta^n+x_j^n)^2} \prod_{k\in I_i\setminus \{j\}} \frac{x_k^n}{\theta^n+x_k^n} , \text{ if $j\in I_i$}.
\end{cases}
\]

Then, the Jacobian matrix of $F$ at the zero vector is the zero matrix. Since $F(\textbf{0})=\textbf{0}$ and the Jacobian matrix of $F(\mathbf{x})-\mathbf{x}$ is $-I$ (the identity matrix), we have that the zero vector is a stable steady state.

We now prove that there is a positive stable steady state. First, note that for any $\theta<K<1$, $F$ converges uniformly to $(1,\ldots,1)\in [K,1]^N$ on the compact set $[K,1]^N$ as $n\rightarrow \infty$. Then, for large but finite $n$, it follows that $F\left([K,1]^N\right)\subseteq [K,1]^N$ and therefore there is a point $\mathbf{x}^*\in[K,1]^N$ such that $F(\mathbf{x}^*)=\mathbf{x}^*$. Second, note that $J(\mathbf{x})$ converges to the zero matrix uniformly on $[K,1]^N$. Thus, the Jacobian matrix of $F(\mathbf{x})-\mathbf{x}$ converges uniformly to $-I$ on the set $[K,1]^N$. Then, it follows that this positive vector, $\mathbf{x}^*$, is a stable steady state. Furthermore, since $\mathbf{x}^*\in [K,1]^N$, Theorem~\ref{thm:previous} implies that $\mathbf{x}^*$ approaches $(1,\ldots,1)$ as $n$ increases. We remark that for this part of the proof we did not need the wiring diagram to be strongly connected (this will be used in Lemma~\ref{lemma:ss_one}).

Now, we will use the fact that the wiring diagram is strongly connected to prove that there are no other stable steady states. Suppose that $\mathbf{x}$ is a stable steady state which by Theorem~\ref{thm:previous} is in $[0,1]^N$. We assume that $\mathbf{x}\notin (\theta,1]^N$ and will show that $\mathbf{x}=(0,\ldots,0)$. Since $\mathbf{x}\notin (\theta,1]^N$, one of the entries, say the $i$-th entry, will be in $[0,\theta)$ and then $\frac{x_i^n}{\theta^n+x_i^n}<0.5$. For any $j$ such that $i\in I_j$, $x_j=F_j(\mathbf{x})=\prod_{k\in I_j} \frac{x_k^n}{\theta^n+x_k^n}<0.5=\theta$ (since one of the factors is $\frac{x_i^n}{\theta^n+x_i^n}$). Then, for any $j$ such that $i\in I_j$, we have $x_j\in[0,\theta)$. Repeating this process, it follows that if there is a path from $i$ to $k$, then $x_k\in[0,\theta)$. Since the wiring diagram is strongly connected, this process covers all vertices of the wiring diagram. Thus, $\mathbf{x}\in[0,\theta)^N$. Since the zero vector is already a stable steady state in $[0,\theta)^N$, using Theorem~\ref{thm:previous} have $\mathbf{x}=(0,\ldots,0)$. This finishes the proof. \hfill $\square$

Using Theorem~\ref{thm:sc} we can explain why Examples \ref{eg:1d} and \ref{eg:2d} had 2 stable steady states only; their wiring diagrams are strongly connected. On the other hand, the wiring diagram of Example~\ref{eg:3d} is not strongly connected, so Theorem~\ref{thm:sc} does not apply. However, we see that the wiring diagram has 2 strongly connected components (maximal strongly connected subgraphs), namely $G_1=\{(1,2),(2,1)\}$ and $G_2=\{(3,3)\}$. 

We now study the number of stable steady states in the case that the wiring diagram is not strongly connected, but is composed of many strongly connected components. 

\begin{example}\label{eg:main}
Consider the conjunctive network with wiring diagram given in Figure~\ref{fig:eg_main}a. The wiring diagram is not strongly connected, but is composed by several strongly connected components as well as edges between them. The graph obtained by replacing each strongly connected component by a vertex, and all edges between 2 strongly connected components by a single edge, forms a partially ordered set (see Figure~\ref{fig:eg_main}b).
\end{example}

\begin{figure}[h]
\centerline{ \hbox{ 
 \includegraphics[width=8cm]{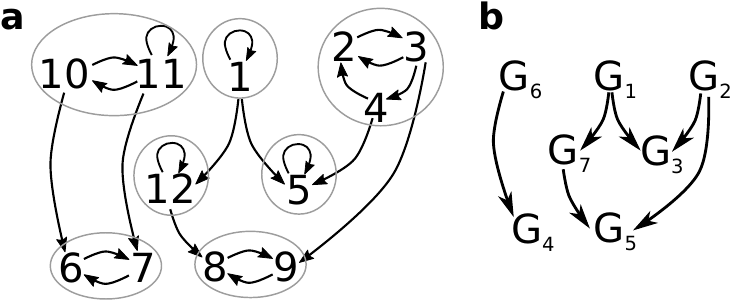}
}}\vspace{-7pt}
 \caption{
Conjunctive network with 12 variables. 
(\textbf{a})
Wiring diagram with strongly connected components indicated by circles. 
(\textbf{b}) 
Connectivity of the strongly connected components results in a partially ordered set.
}
\label{fig:eg_main}
\end{figure}

The main idea behind our approach is that for any stable steady state, $\mathbf{x}$, entries in a strongly connected components are either all zero or all positive. Also, if entries in a strongly connected component are zero, then entries in a strongly connected component ``downstream'' the wiring diagram will also be zero. Indeed, this behavior is seen in Example~\ref{eg:3d}. This means that the connectivity between strongly connected components (Fig.~\ref{fig:eg_main}b) will determine the number of stable steady states. We first need the following lemma.

\begin{lemma}\label{lemma:ss_one}
For any conjunctive network, there is a stable steady state such that each entry is positive and converges to 1 as the Hill coefficient increases.
\end{lemma}
\textbf{Proof.}
Note that the first part of the proof of Theorem~\ref{thm:sc} did not need the wiring diagram to be strongly connected. Thus, there is a positive stable steady state that approaches $(1,\ldots,1)$.
 \hfill $\square$

For simplicity we assume from now on that each strongly connected component has at least one edge so that ``trivial'' strongly connected components are not present. As we will show in the next subsection, such an assumption does not limit our ability to achieve arbitrary number of stable steady states.

{
Before stating the next theorem we need the following definition. Suppose $G_1,G_2,\ldots,G_r$ are the strongly connected components of a wiring diagram. We say that a collection of strongly connected components, $\{G_{k_1},\ldots,G_{k_l}\}$, is an \emph{antichain} if there is no directed path between such components. For example, in Fig.~\ref{fig:eg_main}, $ \{G_2,G_7\}$ is an antichain; $\{G_6,G_7\}$ is an antichain as well. We can see antichains as strongly connected components that do not influence each other. Given an antichain $\{G_{k_1},\ldots,G_{k_l}\}$, we say that $i$ (or the $x_i$ variable) is influenced by this antichain if $i$ can be reached from some node in one of the components in the antichain. For example, for the antichain $\{G_6,G_7\}$, we have that the variables influenced by it are $x_6,x_7,\ldots,x_{12}$. 
}

\begin{theorem}\label{thm:main}
Let $G_1,G_2,\ldots,G_r$ be  the strongly connected components of the wiring diagram of a conjunctive network. Let $\{G_{k_1},\ldots,G_{k_l}\}$ be an antichain. Then, there is a stable steady state $\mathbf{x}^*$ such that $x^*_i=0$ for all variables influenced by the antichain and $x^*_i>0$ for all other $i$.  Furthermore, every stable steady state is of this form.
\end{theorem}
\textbf{Proof.}
Suppose that $G_{k_1},\ldots,G_{k_l}$ are strongly connected components such that there is no edge between them. 

Consider the conjunctive network that results from removing all variables that are influenced by  $G_{k_1},\ldots,G_{k_l}$ (Fig.~\ref{fig:idea_proof_main}). By Lemma~\ref{lemma:ss_one}, this smaller conjunctive network will have a stable steady state with positive entries only, $\mathbf{y}^*$, which converges to $(1,\ldots,1)$. Then, we can use $\mathbf{y}^*$ to construct a stable steady state of the full conjunctive network by completing with zeros. This stable steady state satisfies the required conditions.

To show that every stable steady state satisfies this condition, it is enough to note that if $x_i=0$, then $x_j=0$ for any $j$ such that there is an edge from $i$ to $j$. 
 \hfill $\square$

\begin{figure}[h]
\centerline{ \hbox{ 
 \includegraphics[width=8cm]{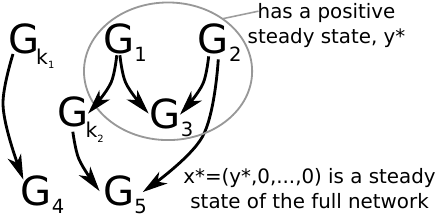}
}}\vspace{-7pt}
 \caption{
Idea for the proof of Theorem~\ref{thm:main} for strongly connected components $G_{k_1}=G_6$ and $G_{k_2}=G_7$. 
First, we remove $G_{k_1}$, $G_{k_2}$, $G_4$ (since there is an incoming edge from $G_{k_1}$), and $G_5$ (since there is an incoming edge from $G_{k_2}$). Then, using Lemma~\ref{lemma:ss_one}, the smaller conjunctive network with strongly connected components $G_1,G_2,G_3$ will have a positive stable steady state. Finally, we construct a stable steady state of the full conjunctive network by using zeros for the entries corresponding to $G_{k_1}$, $G_{k_2}$, $G_4$, and $G_5$.
}
\label{fig:idea_proof_main}
\end{figure}

\textbf{Example \ref{eg:main} (cont.) } If we use Theorem~\ref{thm:main} on the conjunctive network given in Fig.~\ref{fig:eg_main}a considering the strongly connected components $G_6$, $G_7$, we obtain that there is a stable steady state of the form $(x_1^*,x_2^*,x_3^*,x_4^*,x_5^*,0,0,0,0,0,0,0)$ with the first 5 entries positive. Furthermore, this stable steady state converges to $111110000000$ (parenthesis omitted for brevity). With this in mind, from now on we say that a stable steady state $\mathbf{x}^*$ \emph{has the form} $\mathbf{z}\in\{0,1\}^n$ if $sign(x_i)=z_i$ for each $i$, where $sign$ is the sign function. For example, the stable steady state $(x_1^*,x_2^*,x_3^*,x_4^*,x_5^*,0,0,0,0,0,0,0)$ has the form $111110000000$.

Theorem~\ref{thm:main} guarantees that there is a one to one correspondence between stable steady states of a conjunctive network and the antichains of the partially ordered set of strongly connected components. We restate this as a theorem.

\begin{theorem}\label{thm:antichain}
Suppose $G_1,\ldots,G_r$ are the strongly connected components of the wiring diagram of a conjunctive network. Then, there is a one to one correspondence between the stable steady states of the conjunctive network and the antichains of the partially ordered set of strongly connected components. 
\end{theorem} 

\textbf{Example \ref{eg:main} (cont.) } Consider the conjunctive network given in Fig.~\ref{fig:eg_main}a. By complete enumeration we see that the partially order set (Fig.~\ref{fig:eg_main}b) has 30 antichains; namely: 
$ \emptyset$,
$\{G_1\}$, $\{G_2\}$, $\{G_3\}$, $\{G_4\}$, $\{G_5\}$, $\{G_6\}$, $\{G_7\}$,
$\{ G_1,G_2 \}$,  $\{ G_1,G_4 \}$,  $\{ G_1,G_6 \}$,  
$\{ G_2,G_4 \}$, $\{ G_2,G_6 \}$, $\{G_2,G_7\}$, 
$\{G_3,G_4\}$, $\{G_3,G_5\}$, $\{G_3,G_6\}$, $\{G_3,G_7\}$,
$\{G_4,G_5\}$,  $\{G_4,G_7\}$,
$\{G_5,G_6\}$,
$\{G_6,G_7\}$,
$\{G_1,G_2,G_4\}$, $\{G_1,G_2,G_6\}$,
$\{G_2,G_4,G_7\}$, $\{G_2,G_6,G_7\}$,
$\{G_3,G_4,G_5\}$, $\{G_3,G_4,G_7\}$, $\{G_3,G_5,G_6\}$, $\{G_3,G_6,G_7\}$. Thus, Theorem~\ref{thm:antichain} guarantees that the conjunctive network has exactly 30 stable steady states. Furthermore, using the elements of each antichains, we know the form of these stable steady states (see Table~\ref{tab:ss} in Appendix). Note that this result is independent of the Hill coefficient $n$, as long as it is large enough. 

We finish this subsection with a proposition which will simplify the generation of conjunctive networks with arbitrary number of stable steady states.

\begin{proposition}\label{prop:disj}
Suppose that the wiring diagram of $f$ consists of the disjoint union of 2 directed graphs, $W_1$ and $W_2$. Denote with $g$ and $h$ the conjunctive networks corresponding to $W_1$ and $W_2$, respectively. Then, the number of stable steady states of $f$ is the product of the number of stable steady states of $g$ and the number of stable steady states of $h$.
\end{proposition}
\textbf{Proof.}
Given $\mathbf{x}\in [0,1]^N$, denote $\mathbf{x}_{W_1}:=(x_i)|_{i\in W_1}$ and $\mathbf{x}_{W_2}:=(x_i)|_{i\in W_2}$. The proof now follows from the fact that $\mathbf{x}$ is a stable   steady state of $f$ if and only if $\mathbf{x}_{W_1}$ is a stable steady state of $g$ and $\mathbf{x}_{W_2}$ is a stable steady state of $h$.
 \hfill $\square$

\textbf{Example \ref{eg:main} (cont.) } We can use Proposition~\ref{prop:disj} in addition to Theorem~\ref{thm:antichain} to compute the number of stable steady states more efficiently. For example, consider the conjunctive network given in Fig.~\ref{fig:eg_main}a.
The conjunctive network corresponding to the strongly connected components $G_4$ and $G_6$ has 3 stable steady states (since the antichains are $\emptyset$, $\{G_4\}$, and $\{G_6\}$).
The conjunctive network corresponding to the other strongly connected components has 10 stable steady states (since the antichains are $\emptyset$, $\{G_1\}$, $\{G_2\}$, $\{G_3\}$, $\{G_5\}$, $\{G_7\}$, $\{G_1,G_2\}$, $\{G_2,G_7\}$, $\{G_3,G_7\}$, and $\{G_3,G_5\}$). Then, the full conjunctive network has $3\times 10=30$ stable steady states. 

\subsection{Achieving arbitrary number of stable steady states}
\label{sec:anynum}

Theorem~\ref{thm:antichain} allows us to study the dynamical problem of counting the number of stable steady states of conjunctive networks as the problem of counting the number of antichains in a partially ordered set. Therefore, in order to construct a differential equation with a desired number of stable steady states, it is sufficient to find a partially ordered set with that number of antichains. Then, each antichain can be replaced by a strongly connected component when constructing the conjunctive network. Furthermore, each strongly connected component can have a single vertex (with a self loop).

\begin{theorem}\label{thm:exist_conj_net}
For every integer $s\geq 1$, there exists a conjunctive network such that the number of stable steady states is $s$.
\end{theorem}
\textbf{Proof.}
The proof of this theorem is constructive. 

We handle the case $s=1$ separately. It is enough to note that the 1-dimensional differential equation $x_1'=1-x_1$ has a unique stable steady state. This is a conjunctive network because it can be written as $\frac{d x_1}{dt}=\prod_{k\in \emptyset}\frac{x_k^n}{\theta^n+x_k^n}-x_1$.

For $s\geq 2$, note that the partially ordered set given in Fig.~\ref{fig:chain}a has exactly $s$ antichains, namely $\emptyset$, $G_1$, $G_2,\ldots,G_{s-1}$. Then, the conjunctive network with wiring diagram given in Fig.~\ref{fig:chain}b has $s$ stable steady states. This conjunctive network is given by
\[
\frac{dx_1}{dt}=\frac{x_1^n}{\theta^n+x_1^n}-x_1 
\text{ and } \]
\[
\frac{dx_i}{dt}=\frac{x_i^n}{\theta^n+x_i^n}\frac{x_{i-1}^n}{\theta^n+x_{i-1}^n}-x_i \text{ for $2\leq i \leq s-1$}.
\]
 \hfill $\square$

\begin{figure}[h]
\centerline{ \hbox{ 
 \includegraphics[width=8.3cm]{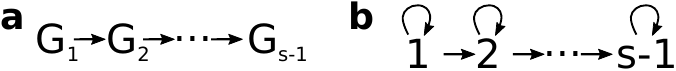}
}}\vspace{-7pt}
 \caption{
Construction of a conjunctive network with $s$ stable steady states. 
\textbf{(a)} Partially ordered set with $s$ antichains.
\textbf{(b)} Wiring diagram corresponding to the partially ordered set.
}
\label{fig:chain}
\end{figure}

\begin{example}
To construct a differential equation with 5 stable steady states it is sufficient to consider 
\begin{align*}
\frac{dx_1}{dt} & =  \frac{x_1^n}{\theta^n+x_1^n}-x_1\\
\frac{dx_2}{dt} & =  \frac{x_2^n}{\theta^n+x_2^n}\frac{x_1^n}{\theta^n+x_1^n}-x_2\\
\frac{dx_3}{dt} & =  \frac{x_3^n}{\theta^n+x_3^n}\frac{x_2^n}{\theta^n+x_2^n}-x_3\\
\frac{dx_4}{dt} & =  \frac{x_4^n}{\theta^n+x_4^n}\frac{x_3^n}{\theta^n+x_3^n}-x_4\\
\end{align*}
The stable steady states of this conjunctive network have the form 0000, 1000, 1100, 1110, and 1111.
\end{example}

\begin{example}\label{eg:s_6}
To construct a differential equation with 6 stable steady states we can use Theorem~\ref{thm:exist_conj_net}, which will require 5 variables. We can also combine Theorem~\ref{thm:exist_conj_net} with Proposition~\ref{prop:disj} to construct a conjunctive network with a smaller number of variables. More precisely, consider
\begin{align*}
\frac{dx_1}{dt} & =  \frac{x_1^n}{\theta^n+x_1^n}-x_1\\
\frac{dx_2}{dt} & =  \frac{x_2^n}{\theta^n+x_2^n}-x_2\\
\frac{dx_3}{dt} & =  \frac{x_3^n}{\theta^n+x_3^n}\frac{x_2^n}{\theta^n+x_2^n}-x_3.
\end{align*}
Since the wiring diagram consists of 2 disjoint graphs, one with 2 and the other with 3 in the corresponding partially ordered sets, then the number of stable steady states is $2\times 3 =6$.
\end{example}

\begin{figure}[h]
\centerline{ \hbox{ 
 \includegraphics[width=8cm]{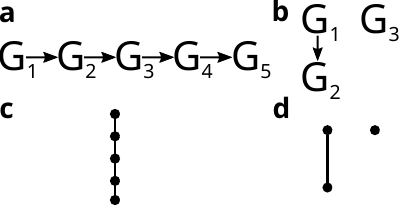}
}}\vspace{-7pt}
 \caption{
Construction of conjunctive network with $6$ stable steady states. 
\textbf{(a)} Partially ordered set with $6$ antichains, which can be used to construct the wiring diagram of a conjunctive network with 6 stable steady states.
\textbf{(b)} Three elements are sufficient to construct a partially ordered set with 6 antichains. Therefore, 3 variables are enough for conjunctive networks to achieve 6 stable steady states.
\textbf{(c,d)} Unlabeled representation of the partially ordered sets above. The direction of the arrows is top to bottom.
}
\label{fig:prod}
\end{figure}

In general, we have the following theorem that improves Theorem~\ref{thm:exist_conj_net} (we omit the easy case $s=1$).  

\begin{theorem}\label{thm:factors}
For every integer $s\geq 2$ with factorization $s=q_1\ldots q_l$, there exists a conjunctive network with $s$ stable steady states and $\sum_{i=1}^l(q_i-1)$ variables. Furthermore, this construction is sharp in the sense that for $s=2^l$, the smallest possible number of variables is $l=\sum_{i=1}^l(2-1)$.
\end{theorem}
\textbf{Proof.}
For the first part of the proof, it is sufficient to note that for each $i=1,\ldots,l$, one can construct a conjunctive network with $q_i$ stable steady states using $q_i-1$ variables. Then, using Proposition~\ref{prop:disj} and induction imply that there is a conjunctive network in $\sum_{i=1}^l(p_i-1)$ variables and that it will have $q_1\ldots q_l=s$ stable steady states.

For the second part of the proof, suppose there is a conjunctive network with $k$ variables and $s=2^l$ stable steady states. Since the wiring diagram has $k$ vertices, then the partially ordered set of strongly connected components has at most $2^k$ antichains. Then, $2^l\leq 2^k$, so $l\leq k$. That is, at least $l$ variables are needed. 
 \hfill $\square$

{

 \subsection{Global long term behavior of conjunctive networks}

The previous theorems tell us about the existence of stable steady states. A potential limitation is that the basin of attractions of these steady states are simply small neighborhoods. Thus, we now provide results about non local behavior that show that the basins of attraction are not small. First we need the following notation.

For $\mathbf{x}\in [0,1]^N$, we can construct an interval along each dimension of $\mathbf{x}$ that corresponds to the qualitative value (low or high) of $x_i$, namely, we construct $[0,\theta-\epsilon]$ if $x_i<\theta$ and $[\theta+\epsilon,1]$ if $x_i>\theta$. The Cartesian product of these intervals results in a hypercube that we denote as $C_\epsilon(\mathbf{x})$. Note that as $\epsilon$ decreases, $C_\epsilon(\mathbf{x})$ approaches one of the $2^N$ regions obtained by cutting $[0,1]^N$ by the $N$ hyperplanes $x_i=\theta$.

With this notation, we state the results in \cite{Veliz:BNODE2012,veliz2014piecewise}  for our purposes. 
\begin{proposition}\label{prop:basin}
If $\mathbf{x}$ is a stable steady state of a conjunctive network, then, for large but finite $n$,  the basin of attraction contains $C_\epsilon(\mathbf{x})$.
\end{proposition}

This proposition indicates that as $n$ increases, the basin of attraction of $\mathbf{x}$ is not simply a small neighborhood of $\mathbf{x}$, but it actually gets closer and closer to containing one of the $2^N$ regions mentioned before.

Since conjunctive networks are cooperative systems ($\frac{\partial F_i}{\partial x_j} \geq 0 $ for $i\neq j$), we also have the following global result \cite{Hirsch1976,Smith1995}.

\begin{theorem}\label{thm:global}
Up to a set of measure zero, every trajectory in a conjunctive network converges to one of its steady states.
\end{theorem}

We remark how the results in this section complement each other. On one hand, Theorem \ref{thm:global} states that trajectories converge to steady states. This by itself doesn't tell us how many steady states there are or what their configuration is. Theorem \ref{thm:antichain} on the other hand tells us exactly what these steady states are. Note that in particular, this means that there is no chaotic or periodic behavior that could potentially disrupt the applicability of our results.

}

\section{Final Remarks: Minimal number of variables}

A simple modification of the proof in Theorem~\ref{thm:factors} results in a lower bound and upper bound for the minimal number of variables needed.
\begin{proposition}\label{prop:bounds}
Let $\mathcal{N}(s)$ denote the minimal number of variables needed to construct a conjunctive network with $s$ stable steady states. Then, for $s=q_1\ldots q_l\geq 2$, we have $\lceil \log_2(s) \rceil \leq \mathcal{N}(s) \leq \sum_{i=1}^l(q_i-1)$, where $\lceil \ \  \rceil$ denotes the ceiling function. These bounds are sharp in the sense that the three quantities are equal for $s=2^l$.
\end{proposition}
\textbf{Proof.}
The inequality $\mathcal{N}(s) \leq \sum_{i=1}^l(q_i-1)$ is simply a restatement of Theorem~\ref{thm:factors}. To prove $\log_2(s) \leq \mathcal{N}(s)$. It is enough to note that since $\mathcal{N}(s)$ is the smallest number of variables, the number of stable steady states, $s$, must satisfy $s\leq 2^{\mathcal{N}(s)}$. 
 \hfill $\square$

We now look at function $\mathcal{N}$ in more detail.
Given $s$, consider the conjunctive network in $N=\lceil \log_2(s) \rceil$ variables given by $\frac{dx_i}{dt} =  \frac{x_i^n}{\theta^n+x_i^n}-x_i$ for $i=1,\ldots,N$. The wiring diagram of this network has $N$ variables, and the corresponding partially ordered set has $2^N$ antichains and therefore the network has $2^N$ stable steady states. Since $2^N \geq s$ and connecting strongly connected components reduces the number of stable steady states (since it reduces the number of  antichains), it could be possible that with $N=\lceil \log_2(s) \rceil$ variables, and with the correct connectivity, one can achieve exactly $s$ stable steady states. This is indeed true for $s=2,3,4,5,6$ as Fig.~\ref{fig:posets_2-6} shows. For these values of $s$, it is enough to have $N=\lceil \log_2(s) \rceil$ variables. However, this is not valid in general. For example, although $\lceil \log_2(7) \rceil =3$, 3 variables are not sufficient to achieve 7 stable steady states. By exhaustive search it can be shown that $\mathcal{N}(7)=4$, Fig.~\ref{fig:posets_7}.

\begin{figure}[h]
\centerline{ \hbox{ 
 \includegraphics[width=8.3cm]{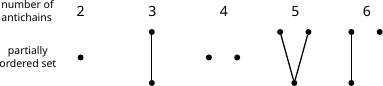}
}}\vspace{-7pt}
 \caption{
Minimal partially ordered sets for  $s=2,3,4,5,6$. 
For these values of $s$, $\mathcal{N}(s)=\lceil \log_2(s) \rceil$. The partially ordered sets can be seen as ordered top to bottom or bottom to top. The partially ordered sets for $s=2,3,$ and $6$ correspond to Examples~\ref{eg:1d}, \ref{eg:3d}, and \ref{eg:s_6}(Fig.~\ref{fig:prod}b), respectively.
}
\label{fig:posets_2-6}
\end{figure} 

\begin{figure}[h]
\centerline{ \hbox{ 
 \includegraphics[width=4cm]{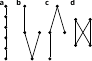}
}}\vspace{-7pt}
 \caption{
Partially ordered sets that result in a network with 7 stable steady states.  (a) Partially ordered set obtained using Theorem~\ref{thm:main}. (b,c,d) Using exhaustive search we find the smallest partially ordered sets that result in a network with 7 stable steady states. Then,  $\mathcal{N}(7)=4$. Note that (b) is the reflection or mirror image of (c) (top-bottom).
}
\label{fig:posets_7}
\end{figure}

In light of Proposition~\ref{prop:disj}, if $s=q_1q_2$ one can find the minimal conjunctive networks to achieve $q_1$ and $q_2$ and use the disjoint union of their wiring diagram to potentially create the minimal conjunctive network needed to achieve $s$ stable steady states. That is, one can hypothesize $\mathcal{N}(q_1q_2)=\mathcal{N}(q_1)+\mathcal{N}(q_2)$. This is indeed true for all composite numbers $s\leq 32$, but not valid in general. For example, since $\mathcal{N}(3)=2$ and $\mathcal{N}(11)=5$, we can construct a conjunctive network with 7 variables that has 33 stable steady states, but 7 is not minimal, since $\mathcal{N}(33)=6$.

\begin{figure}[h]
\centerline{ \hbox{ 
 \includegraphics[width=5cm]{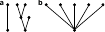}
}}\vspace{-7pt}
 \caption{
Partially ordered sets that result in a network with 33 stable steady states. (a) Using Proposition~\ref{prop:disj} we can consider a network consisting of 2 disconnected networks with 3 and 11 stable steady states using partially ordered sets having 2 and 5 elements, respectively. Then, the full partially ordered set will have 7 elements and the corresponding network will have $3*11=33$ stable steady states. (b) However, 7 variables is not minimal, since one can achieve 33 stable steady states with 6 variables using the partially ordered set shown. Since 6 is minimal, $\mathcal{N}(33)=6$. 
}
\label{fig:posets_33}
\end{figure} 

Computing $\mathcal{N}(s)$ is equivalent to finding the minimal number of elements a partially ordered set so that it has $s$ antichains. This related combinatorial problem has been studied in the context of topologies having $s$ open sets \cite{RAGNARSSON2010138,erne1991counting} (also see \url{oeis.org/A137813}). Although OEIS A137813 contains the value of $\mathcal{N}(s)$ for several values of $s$, it does not include the partially ordered sets, which are needed to construct the conjunctive networks. We used exhaustive search and \cite{poset_website} to find the minimal conjunctive networks that achieve a desired number of stable steady states for $s=2,\ldots,100$ (see link in Appendix for the partially ordered sets).



\section{Conclusion}

Understanding how networks can achieve complex dynamics is a problem that arises in areas such as gene and neural networks. In synthetic biology in particular, it is of interest to be able to achieve dynamical properties such as multistationarity using networks that can be potentially feasible to construct. Although bistability has been achieved using the so-called toggle switch \cite{gardner2000construction}, there is no  methodology that provides concrete networks capable of achieving any desired multistationarity. One difficulty is that it is not trivial to predict the stable steady state behavior of arbitrary differential equations.

A possible solution is to couple simple interactions or modules and use the global structure of the network to achieve multistationarity. AND gates have been studied and designed in the lab, and we showed that by coupling AND gates one can achieve any desired number of stable steady states by choosing the correct network structure. Specifically, we showed that the number of antichains is equal to the number of stable steady states. Furthermore, for any desired number of stable steady states, $s$, we provide a constructive way to design a network with that many stable steady states. This approach also permits to find the minimal number of variables needed to achieve any desired number of stable state states. 

In the context of modularity, our results show that using small modules (single-variable networks with self loops) in very particular combinations (determined by the partially ordered set), one can achieve complex behavior such as multistationarity. 
{
Furthermore, combining our results with existing results about cooperative systems, we can guarantee that (up to a set of measure zero) all trajectories converge to the steady states predicted by our results. 
}

\bibliographystyle{abbrv}

\section*{Appendix}

\begin{table}[h]
\label{tab:ss}
\begin{tabular}{cc}
Antichain & Form of the stable steady state\\
$ \emptyset$ & 1\ 111\ 1\ 11\ 11\ 11\ 1\\
$\{G_1\}$ &  0\ 111\ 0\ 11\ 00\ 11\ 0\\
$\{G_2\}$ &  1\ 000\ 0\ 11\ 00\ 11\ 1\\ 
$\{G_3\}$ &  1\ 111\ 0\ 11\ 11\ 11\ 1\\ 
$\{G_4\}$ &  1\ 111\ 1\ 00\ 11\ 11\ 1\\ 
$\{G_5\}$ &  1\ 111\ 1\ 11\ 00\ 11\ 1\\ 
$\{G_6\}$ &  1\ 111\ 1\ 00\ 11\ 00\ 1\\ 
$\{G_7\}$ &        1\ 111\ 1\ 11\ 00\ 11\ 0\\
$\{ G_1,G_2 \}$ &  0\ 000\ 0\ 11\ 00\ 11\ 0\\  
$\{ G_1,G_4 \}$ &  0\ 111\ 0\ 00\ 00\ 11\ 0\\  
$\{ G_1,G_6 \}$ &  0\ 111\ 0\ 00\ 00\ 00\ 0\\  
$\{ G_2,G_4 \}$ &    1\ 000\ 0\ 00\ 00\ 11\ 1\\ 
$\{ G_2,G_6 \}$ &    1\ 000\ 0\ 00\ 00\ 00\ 1\\ 
$\{G_2,G_7\}$ &      1\ 000\ 0\ 11\ 00\ 11\ 0\\ 
$\{G_3,G_4\}$ &      1\ 111\ 0\ 00\ 11\ 11\ 1\\ 
$\{G_3,G_5\}$ &      1\ 111\ 0\ 11\ 00\ 11\ 1\\ 
$\{G_3,G_6\}$ &      1\ 111\ 0\ 00\ 11\ 00\ 1\\ 
$\{G_3,G_7\}$ &      1\ 111\ 0\ 11\ 00\ 11\ 0\\
$\{G_4,G_5\}$ &      1\ 111\ 1\ 00\ 00\ 11\ 1\\  
$\{G_4,G_7\}$ &      1\ 111\ 1\ 00\ 00\ 11\ 0\\
$\{G_5,G_6\}$ &      1\ 111\ 1\ 00\ 00\ 00\ 1\\
$\{G_6,G_7\}$ &      1\ 111\ 1\ 00\ 00\ 00\ 0\\
$\{G_1,G_2,G_4\}$ &  0\ 000\ 0\ 00\ 00\ 11\ 0\\ 
$\{G_1,G_2,G_6\}$ &  0\ 000\ 0\ 00\ 00\ 00\ 0\\ 
$\{G_2,G_4,G_7\}$ &  1\ 000\ 0\ 00\ 00\ 11\ 0\\ 
$\{G_2,G_6,G_7\}$ &  1\ 000\ 0\ 00\ 00\ 00\ 0\\
$\{G_3,G_4,G_5\}$ &  1\ 111\ 0\ 00\ 00\ 11\ 1\\ 
$\{G_3,G_4,G_7\}$ &  1\ 111\ 0\ 00\ 00\ 11\ 0\\ 
$\{G_3,G_5,G_6\}$ &  1\ 111\ 0\ 00\ 00\ 00\ 1\\ 
$\{G_3,G_6,G_7\}$ &  1\ 111\ 0\ 00\ 00\ 00\ 0
\end{tabular}
\caption{stable steady states of conjunctive network given in Fig.~\ref{fig:eg_main}. The spacing between numbers indicates different strongly connected components.}
\end{table}

The list of the minimal partially ordered sets  is available at \url{github.com/alanavc/minimal-posets}.

\end{document}